\documentclass[10pt]{llncs}

\usepackage{amsmath}

\usepackage{verbatim}
\usepackage[boxed,linesnumbered]{algorithm2e}

\begin{document}

\title{Range Non-Overlapping Indexing}
\author{Hagai Cohen \and Ely Porat\thanks{This work was supported by BSF and ISF}}
\institute{Department of Computer Science, Bar-Ilan University, 52900 Ramat-Gan, Israel
\email{\{cohenh5,porately\}@cs.biu.ac.il}}

\maketitle

\begin{abstract}
We study the \emph{non-overlapping indexing} problem:
Given a text $T$, preprocess it in order to answer queries of the form:
given a pattern $P$, report the maximal set of non-overlapping occurrences of $P$ in $T$.
A generalization of this problem is the \emph{range non-overlapping indexing} where in addition we are given
two indexes $i,j$ to report the maximal set of non-overlapping occurrences between these two indexes.
We suggest new solutions for these problems.
For the non-overlapping problem our solution uses $O(n)$ space with query time of $O(m + occ_{NO})$.
For the range non-overlapping problem we propose a solution with $O(n\log^\epsilon n)$ space
for some $0<\epsilon<1$ and $O(m + \log\log n + occ_{ij,NO})$ query time.
\end{abstract}

\section{Introduction and Related Work}

Given a text $T$ of length $n$ over an alphabet $\Sigma$, the \emph{text indexing} problem is to
build an index on $T$ which can answer pattern matching queries efficiently:
Given a pattern $P$ of length $m$, we want to report all its occurrences in $T$.
There are some known solutions for this problem.
For instance, the \emph{suffix tree}, proposed by Weiner \cite{weiner73},
which is a compacted trie storing all suffixes of the text.
A suffix tree for text $T$ of length $n$ requires $O(n)$ space and can be built in $O(n)$ preprocessing time.
It has query time of $O(m + occ)$ where $occ$ is the number of occurrences of $P$ in $T$.

Range text indexing, also known as \emph{position restricted substring searching},
is the problem of finding a pattern $P$
in a substring of the text $T$ between two given positions $i, j$.
A solution for this problem was presented by Makinen and Navarro \cite{navarro06}.
It uses $O(n\log^\epsilon n)$ space and has query time of $O(m + \log\log n + occ)$.
Their solution is based on another problem - the range searching problem.

The \emph{range searching} problem is to preprocess a set of points in a $d$-dimensional space
for answering queries about the set of points which are contained within a specific range.
Alstrup et al \cite{alstrup00} proposed a solution for the orthogonal two dimensional
range searching problem when all the points are from $n\times n$ grid, which costs $O(n\log^\epsilon n)$ space
and has $O(\log\log n + k)$ query time, where $k$ is the number of points inside the range.
Grossi and Iwona in \cite{grossi05} have shown how to use Alstrup's data structure
to get all the points inside a particular range in a specific order using some kind of rank function.

In text indexing we are sometimes interested in reporting only the non-overlapping occurrences of $P$ in $T$.
There is such interest in fields such as pattern recognition, computational linguistics, speech recognition,
data compression, etc.
For instance, we might want to compress a text by replacing each non-overlapping occurrence of a substring of it with a pointer to a single copy of the substring.

Another problem is the \emph{string statistics} problem \cite{apostolico96,brodal02} which consists of preprocessing
a text $T$ such that when given a query pattern $P$, the maximum number of non-overlapping occurrences of $P$ in $T$ can be reported efficiently.
However, in the string statistics problem we only return the number of non-overlapping occurrences not the actual occurrences.
In this paper, we present the first non-trivial solution for the \emph{non-overlapping indexing} problem where we want to report the maximal sequence of non-overlapping occurrences of $P$ in $T$.

Keller et al \cite{keller07} proposed a solution for a generalization of this problem called
the \emph{range non-overlapping indexing} where we want to report the non-overlapping occurrences in a substring of $T$.
Their solution has query time of $O(m + occ_{ij,NO}\log\log n)$ and uses $O(n\log n)$ space,
where $occ_{ij,NO}$ is the number of the maximal non-overlapping occurrences in the substring $T[i:j]$.

Crochemore et al \cite{Crochemore08} suggested another solution for the \emph{range non-overlapping indexing} problem.
Their solution has optimal query time of $O(m + occ_{ij,NO})$ but requires $O(n^{1 + \epsilon})$ space.

In this paper, we present new solutions for the \emph{non-overlapping indexing} problem,
which use the periodicity of the text and pattern in order to minimize the query time.
Our solution for \emph{non-overlapping indexing} problem uses $O(n)$ space with optimal query time of $O(m + occ_{NO})$.
For the \emph{range non-overlapping indexing} problem we present a solution of $O(n\log^\epsilon n)$ space for some $0<\epsilon<1$ with $O(m + \log\log n + occ_{ij,NO})$ query time.

\section{Preliminaries}

Let $n$ be the length of the text $T$. And let $m$ be the length of the pattern $P$.
For two integers $i, j$ ($i\leq j$), $T[i:j]$ is the substring of $T$ from $i$ to $j$.

We will use the Suffix Tree as our main data structure.
Each leaf in the Suffix Tree represents a suffix in the text.
In each leaf we save two values:
$y$ - the start location of its suffix in the text
and $x$ - the location of the leaf in a left to right order of all
the leaves of the Suffix Tree (lexicographic order, for example).
We have two orders on the leaves and therefore on the suffixes as well: x-order and y-order.
The y-order is the text order and the x-order is the suffix tree leaves order.

When we search for a pattern $P$ in a Suffix Tree ST of $T$, we finish searching at some node $v$.
The subtree rooted by $v$ has all the occurrences of $P$ in $T$ in its leaves.
We denote by $l$ and $r$ the x-value of the leftmost leaf and the x-value of the rightmost leaf of that subtree respectively.
Therefore, the occurrences of $P$ in $T$ are all the leaves with x-value between $l$ and $r$.

We denote the number of occurrences of $P$ in $T$ by $occ$.
The number of non-overlapping occurrences will be denoted as $occ_{NO}$.
The number of occurrences of $P$ in $T[i:j]$ and the non-overlapping occurrences of $P$ in $T[i:j]$ will be denoted by $occ_{ij}$ and $occ_{ij,NO}$ respectively.

\section{A Solution for Non-Overlapping Indexing}

We use a new approach for solving this problem.
Our solution uses the periodicity of the text and the pattern.
We divide patterns for two types: \textbf{periodic} and \textbf{aperiodic}.
A different strategy will be used for each type.

\begin{definition}
A pattern that can appear more than twice overlapping is called a \textbf{periodic pattern}.
A pattern that can appear at most twice overlapping is called an \textbf{aperiodic pattern}.
\end{definition}

\subsection{Aperiodic Pattern}
In the aperiodic case we use the periodicity of the pattern to answer a query.
We use the familiar Suffix Tree to get all the leaves that correspond to the given pattern.
After we have all the leaves, we need to remove the overlapping occurrences.
This can be done by sorting the leaves in y-order, going over the sorted list and filtering the overlapping occurrences.
However, sorting $occ$ items costs $O(occ\log occ)$ which is greater than the optimal $O(occ)$.
In order to solve this sorting part we use the following theorem.

\begin{theorem}\label{teo:renaming}
All occurrences of a pattern can be reported and sorted in text order in $O(m + occ)$ time using $O(n)$ space.
\end{theorem}

\begin{proof}
We use a Suffix Tree to get all the occurrences in $O(m + occ)$ time.
For sorting all the occurrences we will use a renaming method on the Suffix Tree.

Each leaf has its location, i.e., its y-value index, in the whole tree.
Saving this location for a leaf costs $\log n$ bits because the whole tree has $n$ leaves.
Hence, the domain for the location is $n$.
Nevertheless, we are interested in the order of the leaves in a subtree of the occurrences and not in the whole tree.
Thus, we would like to save the location of a leaf for a subtree with less leaves.
If we save for each leaf its location in a subtree with less leaves,
for example $\sqrt{n}$ leaves, it will cost us only $\log\sqrt{n}$.
Therefore, for each leaf, aside from keeping its location in the whole suffix tree,
we save its location in a subtree of size $\sqrt{n}$, its location in a subtree of size $\sqrt[4]{n}$,
and so on for all subtrees of size $\sqrt[2^i]{n}$ for $i \geq 1$ until we reach a constant size.

We use Radix Sort which can sort $n$ numbers in a domain of $n^2$ in $O(n)$ time for sorting the leaves by their locations.
Given a subtree whose leaves we wish to sort in y-order, we can sort them by the locations of the subtree of size at most $O(occ^2)$, this will cost us only $O(occ)$ by using Radix Sort because we sort $occ$ items in a domain of at most $occ^2$.\\*
In each leaf we save:
\begin{enumerate}
    \item[] $\log n$ bits for its location in the ST.
    \item[] $\log\sqrt{n}$ bits for its location in the subtree of size $\sqrt{n}$.
    \item[] $\log\sqrt[4]{n}$ bits for its location in the subtree of size $\sqrt[4]{n}$.
    \item[] etc.
\end{enumerate}
This sums as following:
$\log n + \log\sqrt{n} + \log\sqrt[4]{n} + \log\sqrt[8]{n}... = \log n + \frac{1}{2}\log n + \frac{1}{4}\log n + \frac{1}{8}\log n +... \leq 2\log n$.

Therefore, we save only $2\log n$ bits per leaf. We have $n$ leaves summing up to $n\cdot 2\log n = O(n\log n)$ bits which is $O(n)$ space.
\qed
\end{proof}

Theorem~\ref{teo:renaming} provides us a sorted list of all occurrences in $O(m + occ)$ query time and $O(n)$ space.
By filtering the overlapping occurrences which costs $O(occ)$ time, we are through.
Because in an aperiodic pattern, $O(occ)$ = $O(occ_{NO})$, the query time equals to $O(m + occ_{NO})$.

\subsection{Periodic Pattern}
The periodic case is more complex.
In this case we use the periodicity of the text in order to answer a query.

\begin{definition}
A node in the Suffix Tree which represents a suffix which is a periodic pattern is called \textbf{period node}.
\end{definition}

\begin{definition}
Let $s$ be a string. We define a \textbf{period} of $s$ to be a string $p$, such that $s=p^{t}\acute{p}$, for $t \geq 1$ where $\acute{p}$ is prefix of $p$.
\end{definition}

\begin{lemma}
A period node has only one son which can also be a period node.
\end{lemma}

\begin{proof}
Let $a$ be a period node. Therefore, the string represented by $a$ has a period $p$.
For a son of $a$ to be a period node too it must continue the period $p$.
If $a$ ends with a character $c$ than the node which has the next character in $p$ that is after that $c$ is the period node.
There can be only one child of $a$ which can start with this character.
Thus, a period node can have only one son which is also period node.
\qed
\end{proof}

Note that by the period definition a string can have more than one period,
by taking $p_{new}=pp$ for example.
Nevertheless, each period must overlap with a shorter period of the same string.
Thus, there can't be more than one such character $c$.

\begin{definition}
The path that starts from the a period node and goes through all nodes
which continue that period in the Suffix Tree is called a \textbf{period path}.
We denote the number of nodes in a period path to be the period path length.
\end{definition}

\begin{lemma}\label{lem:periodPathLen}
Let $p_{2}$ be a period path on a path $PT$ to the root in the Suffix Tree.
Than $p_{2}$ must be at least twice as long as the previous period path $p_{1}$ on $PT$.
\end{lemma}

\begin{proof}
On $PT$, between $p_{1}$ and $p_{2}$ there must be at least one node which is not a period node.
For $p_{2}$ to be a period path is must represent a period suffix which must be started with the period of $p_{1}$.
Moreover, $p_{2}$ period suffix should be continued by the character
of the next node after $p_{1}$ which is not a period node.
After it there must be the period of $p_{1}$ again.
Therefore, $p_{2}$ length is at least twice longer than $p_{1}$.
\qed
\end{proof}

\begin{lemma}\label{lem:periodPathNum}
The largest number of different period paths contained in the path from the root to a period node is $\log n$.
\end{lemma}


\begin{proof}
Let $PT$ be the path from the root to some period node in the Suffix Tree.
According to Lemma ~\ref{lem:periodPathLen} each period path on $PT$
must be at least twice as long as the previous period path on $PT$.
Therefore, if we have more than $\log n$ period paths on $PT$,
than the length of the last period path must be greater than $n$ which is the text length.
Hence, the number of period paths on the same path can't be more than $\log n$.
\qed
\end{proof}

\begin{definition}
A \textbf{period sequence} is the maximum substring in the text of some period which is repeated more than twice.
We will mark it as $[s,e]$, where $s$ and $e$ are the start index and the end index of the period sequence accordingly.
The period sequences of a period pattern are all the period sequences which start with that periodic pattern.
The \textbf{period length} of a period sequence is the length of the period inside repeated the sequence.
\end{definition}

\begin{example}
Lets $T$ be the text ``abababcabababcabababc''.
The period sequences are:
For the period ``ab'', the period sequences are [1,6], [8,13], [15,20] which have period length $2$.
For the period ``abababc'', the period sequence is [1,21] which has period length $7$.
\end{example}

\begin{lemma}\label{lem:periodSequenceToOcc}
Given a list $L$ of all the period sequences of some periodic pattern in the text.
All the non-overlapping occurrences of that periodic pattern can be retrieved from this list in $O(occ_{NO})$ time.
\end{lemma}

\begin{proof}
For a period sequence $[s,e]$ with period length $pl$, the non-overlapping occurrences of a
periodic pattern $P$ with length $m$ are the group:
$s + i*step | step = \lceil \frac{m}{pl} \rceil * pl, 0\leq i \leq \frac{e-s-1}{step}$
which can be easily calculated.

We report all the occurrences in each period sequence in $L$.
The number of all occurrences we report is $O(occ_{NO})$, so the total time for reporting all the non-overlapping occurrences from $L$ is $O(occ_{NO})$.
\qed
\end{proof}

For answering the non-overlapping indexing we use the following data structure.
We build a data structure for each period path on the Suffix Tree,
saving a list of all period sequences sorted by their length for each period path.
Each period node in the Suffix Tree is on a period path.
We save a pointer from each period node on the Suffix Tree to the period sequence list of its period path.
This pointer will point to the period node appropriate length on the period sequences list.

\begin{theorem}\label{teo:NLOGNSolution}
Using the data structure described above, all period sequences of a period pattern
 can be retrieved in $O(m + occ_{NO})$ time using $O(n\log n)$ space.
\end{theorem}

\begin{proof}
On a query we go to that data structure, get all the period sequences and by Lemma~\ref{lem:periodSequenceToOcc}
we calculate all the non-overlapping occurrences.
This will take $O(occ_{NO})$ time, because the number of the period sequences that we will get is less than the number of the non-overlapping occurrences of the pattern.
If we come up with a long pattern we won't get shorter period sequences which don't fit the pattern
so we won't get unnecessary period sequences.

The space for this data structure is $O(n\log n)$.
This is because there are $n$ nodes where each one can be at most in $\log n$ period paths.
For each node, we save all its period paths so it needs $O(n\log n)$.
\qed
\end{proof}

Now, we will show how to reduce the space needed for this data structure.
\begin{definition}
Let's define a \textbf{degree of a period sequence} to be the maximum degree of a period sequence included in it plus one.
A period sequence without any period sequences in it will has the degree $0$.
\end{definition}

\begin{lemma}\label{lem:patternPeriodNum}
The maximum degree of a period sequence can be at most $\log n$.
\end{lemma}

\begin{proof}
Let $ps$ be the period sequence with the maximum degree in the text.
In $ps$ there is a period sequence with a degree decreased by one.
In that period sequence there is another period sequence with a degree decreased by one.
And so on until we receive a period sequence $ps_{0}$ with the degree $0$.
The length of each period sequence from $ps_{0}$ to $ps$ is at least twice the length of the period sequence in it.
The maximum length of $ps$ can be at most $n$, therefore, its degree can be at most $\log n$.
\qed
\end{proof}

\begin{lemma}\label{lem:allPeriodNum}
There are at most $O(n)$ period sequences.
\end{lemma}

\begin{proof}
We will count the number of period sequences in each degree:
\begin{itemize}
    \item[] The number of period sequences of degree $0$ can be at most $n$.
    \item[] The number of period sequences of degree $1$ can be at most $\frac{n}{2}$.
    \item[] The number of period sequences of degree $2$ can be at most $\frac{n}{4}$.
    \item[] \ldots
    \item[] The number of period sequences of degree $\log n$ can be at most $1$.
\end{itemize}
Summery: $n + \frac{n}{2} + .. + 1 \leq 2n = O(n)$
\qed
\end{proof}

\begin{theorem}
The data structure in Theorem~\ref{teo:NLOGNSolution} can be saved using only $O(n)$ space.
\end{theorem}

\begin{proof}
We save all the period paths in a data structure. Each one with its own period sequences.
Each period sequence appears in only one period path.
From Lemma~\ref{lem:allPeriodNum} we have $O(n)$ period sequences.
Therefore we save at most $O(n)$ space for all the period sequences.
Thus, we need only $O(n)$ space for this data structure.
\qed
\end{proof}

\begin{corollary}
Using these two different strategies for each type of pattern we can solve
the non-overlapping indexing problem in $O(n)$ space with $O(m + occ_{NO})$ query time.
\end{corollary}

\section{A solution for Range Non-Overlapping Indexing}

We propose a better solution for this problem.
Our solution costs $O(n\log^\epsilon n)$ space for some $0<\epsilon<1$
and has query time of $O(m + \log\log n + occ_{ij,NO})$.

\subsection{Rank Sensitive Range Searching}
We use a data structure for answering the two-dimensional orthogonal range searching problem.
Alstrup et al \cite{alstrup00} proposed a data structure for this problem which
requires $O(n\log^\epsilon n)$ space with query time of $O(\log\log n + k)$
where $k$ is the number of points in the range.

Nevertheless, this range query data structure reports all the points in the range with no specific order.
We want to get those points in a specific order.
Therefore, in addition to this data structure we will use a method suggested by Grossi et al \cite{grossi05}
for a rank sensitive data structure. This gives us a data structure which uses $O(n\log^\epsilon n)$ space
with query time of $O(\log\log n)$ and  $O(1)$ per point, where the points will be reported in rank order.
For simplicity we will call this data structure RSDS from now on.

\subsection{Aperiodic Pattern}
We use a Rank Sensitive Data Structure to answer aperiodic queries.
In the RSDS we store all the occurrences as points by their $x$ value and $y$ value,
where the rank function of a point will be its $y$ value.
Given a pattern $P$ and range $i,j$ we can get $l,r$ from the Suffix Tree,
the leftmost leaf and the rightmost leaf which are occurrences of $P$.
Then we will do a range query for points within $[i,j]x[l,r]$ to get all the correct occurrences.
Because the rank in the RSDS is by $y$ value, we will get the points and therefore the occurrences,
sorted in the text order.
The only remaining action is filtering the overlapping occurrences.

The RSDS costs $O(n\log^\epsilon n)$ space. RSDS query time is $O(\log\log n + k)$
where $k$ is the size of the output which is equal to $O(occ_{ij})$.
In our case $k$ equals $O(occ_{ij,NO})$ because for an aperiodic pattern $O(occ_{ij}) = O(occ_{ij,NO})$.
Therefore, aperiodic pattern has query time of $O(m)$ for searching the Suffix Tree plus $O(\log\log n + occ_{ij,NO})$ for the RSDS query.
Concluded in $O(m + \log\log n + occ_{ij,NO})$.

\subsection{Periodic Pattern}
The periodic case is more complex.
We save all period sequences in the text as points in two RSDS.
For a periodic sequence $[s,e]$ with period length $pl$, we save two points $(x_{1},y_{1}), (x_{2},y_{2})$ with the following values:
\begin{itemize}
  \item[] $x_{1}$ = Index of the suffix of $s$ in the left to right order of all the ST leaves
  \item[] $y_{1}$ = $s$
  \item[] $x_{2}$ = Index of the suffix of $s$ in the left to right order of all the ST leaves
  \item[] $y_{2}$ = $e-pl + 1$
\end{itemize}

Point $(x_{1}, y_{1})$ will be saved in the first RSDS with a rank function of $x$ value in descending order.
Point $(x_{2}, y_{2})$ will be saved in the second RSDS with a rank function of $x$ value in ascending order.
Sometimes there will be multiple period sequences with the same start index or end index, each with a different degree.
When this happens we save only the one with the highest degree.
We can easily convert a period sequence $[s,e]$ to these two points and vice versa.

Following Lemma~\ref{lem:allPeriodNum} the number of points in the two RSDS is $O(n)$.
Hence, each RSDS costs $O(n\log^\epsilon n)$ space.

Given a pattern $P$ of length $m$ and range $i,j$ we answer using Algorithm~\ref{algo:PeriodicPatternRangeQuery}.

\begin{algorithm}[H]
\caption{Periodic Pattern Range Query}
\label{algo:PeriodicPatternRangeQuery}
    Get the range $l,r$ from the Suffix Tree \;
    $S \longleftarrow $ Query first RSDS for all points within $[i,j]x[l,r]$ \;
    $S \longleftarrow S~\cup$ Query second RSDS for all points within $[i,j]x[l,r]$ \;
    $S \longleftarrow S~\cup$ Query first RSDS for the first point within $[1,i]x[l,r]$ \;
    $PS \longleftarrow convertAllPointsToPeriodSequences(S)$ \;
    $PS2 \longleftarrow \emptyset$ \;
    \For{$ps \in PS$} {
        $x \longleftarrow ps$ \;
        \While{$x$ period length is greater than $m$} {
            $x \longleftarrow $ the first period sequence inside $x$ \;
        }
        $PS2 \longleftarrow PS2 \cup \{x\}$
    }
\end{algorithm}

Getting the first period sequence inside a period sequence can be done by using another data structure saving for each period sequence its period length, and a pointer to the first period sequence in it.
Thus, given a period sequence it costs $O(1)$ to find the first period sequence inside it with a degree decreased by one.

\begin{lemma}\label{lem:periodSequenceDegrees}
The number of period sequences we have to go down in order to find our appropriate period sequence
is lower than the number of occurrences that will be extracted from the period sequences
inside the first period sequence we received.
\end{lemma}

\begin{proof}
Each degree we get down means that there is another occurrence in the next period sequence.
Each step down, adds at least another occurrence.
Therefore, until we get to the appropriate period sequence we work at most $O(k)$
where $k$ is the number of occurrences we will get from the period sequences inside the first period sequence we encounter.
\qed
\end{proof}

By Lemma~\ref{lem:periodSequenceDegrees} it does not cost us more time
when we get a period sequence whose period length is longer than the pattern length.

\begin{theorem}
All the non-overlapping occurrences can be calculated from the period sequences got by these three queries.
\end{theorem}

\begin{proof}
First of all we will see that each period sequence we get from these queries has at least one occurrence of $P$.
Let $(x,y)$ be the point we get. It corresponds to a period sequence $[s,e]$.
We get only points which have $x$ values between $l$ and $r$.
The $x$ value of a point is the index of the suffix of $s$ in the left to right order of all the ST leaves.
So if we get a point $(x,y)$ with $x$ between $l$ and $r$ it means that
the corresponding period sequence $[s,e]$ has an occurrence of $P$.
This happens because all the leaves in the ST between $l$ and $r$ are occurrences of $P$.

Now, we need to prove two more things.
The first is that all the period sequences we get from the RSDS are suitable for us
and that we haven't got unnecessary period sequences, which don't fit to $P$ in the range $i,j$.
The second thing is that we didn't miss any period sequence which can have some suitable occurrences.

We start by proving that we get all the occurrences of $P$ in the range $[i,j]$
from the period sequences we get in the three queries.
Period sequences of $P$ in $T$ can be in some cases.
Let $[s,e]$ be our period sequence.

The first case is that $[s,e]$ is out of the range $[i,j]$, $s<e\leq i<j$ or $i<j\leq s<e$.
In this case we wouldn't like to get this period sequences at all.
The first two queries will not resolve these period sequences because we do a query on $[i,j]x[l,r]$
but $s$ is out of range and the points corresponding to this period sequence have $y$ value equals $s$.
Nevertheless, we can get this period sequence in the third query. However, it will be at most one period sequence
which can be checked in $O(1)$ time.

The second case, which is the simplest, is that $[s,e]$ is fully inside the range $[i,j]$, $i\leq s<e\leq j$.
In this case we get all the suitable period sequences from the first query.
Nevertheless, we can get the same period sequence twice, first from the first query and again from the second query.
Therefore, we will have to check any period sequence that we get in order to prevent duplicate occurrence reporting.

The third case is when only $e$ or $s$ is out of range but not both, $s<i<e\leq j$ or $i\leq s<j<e$.
This time if $i\leq s<j<e$ we will get the period sequence from the first query.
Otherwise, if $s<i<e\leq j$ we will get the period sequence from the second query.

The fourth case is when the range $[i,j]$ is fully inside the period sequence $[s,e]$, $s<i<j<e$.
In order to solve this case we have the third query which will resolve the last start of a period sequence
which is before index $i$.
This period sequence can be checked in $O(1)$ time.
\qed
\end{proof}


\begin{corollary}
Using these two different strategies for each type of pattern, the range non-overlapping text indexing problem
can be solved in $O(n\log^\epsilon n)$ space for some $0<\epsilon<1$ and query time of $O(m + \log\log n + occ_{ij,NO})$.
\end{corollary}

\section{Conclusion}
We have studied the problem of non-overlapping indexing.
In this paper, we provide the first non-trivial solution for this problem.
In addition we proposed a better solution for a generalization of this problem, the range non-overlapping problem.

\bibliographystyle{splncs}
\bibliography{nonOverlapBib}

\end{document}